\documentclass[a4paper]{article}

\usepackage{times}

\newcommand{\seeAppendix}[1]{in Section~\ref{#1} of the Appendix}

\usepackage[pdftex]{graphicx}
\usepackage[pdftex]{hyperref}
\usepackage{amsmath,amssymb,amsthm}

\theoremstyle{plain}
\newtheorem{theorem}{Theorem}
\newtheorem{lemma}[theorem]{Lemma}
\newtheorem{corollary}[theorem]{Corollary}
\theoremstyle{definition}
\newtheorem{Definition}[theorem]{Definition}
\theoremstyle{remark}
\newtheorem{example}[theorem]{Example}
\newtheorem{remark}[theorem]{Remark}

\usepackage{triotex}
\usepackage{xcolor}
\usepackage{listings}
\usepackage{array}
\usepackage{fancyhdr}
\usepackage{xspace}

\newcommand{\limpl}{\Rightarrow}
\newcommand{\liff}{\Leftrightarrow}

\usepackage{tikz}
\usetikzlibrary{arrows,automata,shapes.symbols,positioning}

\usepackage{subcaption}

\newcommand{\emk}{\mathop{\$}}
\newcommand{\Sigmaend}{\Sigma_{\emk}}

\newcommand{\powerset}[1]{\wp(#1)}

\newcommand{\Orm}{\mathrm{O}}

\usepackage{eufrak}
\newcommand{\async}[1]{\ensuremath{\mathfrak{A}{\ifthenelse{\not \equal{#1}{}}{_{#1}}{}}}}
\newcommand{\algo}[1]{\ensuremath{\mathfrak{I}{\ifthenelse{\not \equal{#1}{}}{({#1})}{}}}}

\newcommand{\valc}[1]{\ensuremath{\mathrm{ACC}(#1)}}

\newcommand{\fakepar}[1]{\vskip 1pt\textbf{#1}}

\newcommand{\lang}[1]{\mathcal{L}(#1)}

\newcommand{\allmodels}[2]{\ldoubsq{#1}\rdoubsq_{#2}}

\newcommand{\pad}{\Box}

\newcommand{\cat}{\mathbin{\circ}}

\newcommand{\comp}[1]{\overline{#1}}

\newcommand{\aut}[1]{\mathcal{A}_{\nref{#1}}}
\newcommand{\autd}[1]{\mathcal{A}_{#1}}

\newcommand{\rest}{\mathit{rest}}
\newcommand{\Result}{\ensuremath{\text{\lstinline|Result|}}}
\newcommand{\len}{\mathit{len}}
\newcommand{\dec}{\mathit{dec}}

\newcommand{\vc}{\mathsf{vc}}
\newcommand{\cvc}{\widetilde{\vc}}
\newcommand{\ce}{\mathsf{ice}}

\begin{document}

\title{Asynchronous Multi-Tape Automata Intersection: \\
Undecidability and Approximation}

\author{Carlo A. Furia \\
ETH Zurich, Switzerland \\
\url{caf@inf.ethz.ch}
}

\maketitle

\begin{abstract}
  When their reading heads are allowed to move completely
  asyn\-chro\-nous\-ly, finite-state automata with multiple tapes
  achieve a significant expressive power, but also lose useful closure
  properties---closure under intersection, in particular.  This paper
  investigates to what extent it is still feasible to use multi-tape
  automata as recognizers of polyadic predicates on words.  On the
  negative side, determining whether the intersection of asynchronous
  multi-tape automata is expressible is not even semidecidable.  On
  the positive side, we present an algorithm that computes
  under-approximations of the intersection; and discuss simple conditions
  under which it can construct complete intersections.  A prototype
  implementation and a few non-trivial examples demonstrate the
  algorithm in practice.
\end{abstract}

\section{Automata As Decision Procedures} \label{sec:introduction}
Standard finite-state automata are simple computing devices widely used in computer science.
They define a robust class of language acceptors, as each automaton instance $A$ identifies a set $\lang{A}$ of words  % over a finite alphabet 
that it accepts as input.
The connection between finite-state automata and predicate logic has been well-known since the work of B\"uchi~\cite{Bu60,Bu62} and others~\cite{Tr58,El61}, and is widely used in applications such as model-checking: each automaton $A_P$ can be seen as \emph{implementing} a monadic (that is, unary) predicate $P(x)$, in the sense that the set  $\lang{A_P}$ of words accepted by the automaton corresponds to the set $\{x \mid x \models P(x) \}$ of models of the predicate.
Logic connectives (negation $\neg$, conjunction $\land$, etc.) translate into composition operations on automata (complement, intersection $\cap$, etc.), so that finite-state automata can capture the semantics of arbitrary first-order monadic formulas whose interpreted atomic predicates are implementable.
This gives a very efficient way to decide the satisfiability of monadic logic formulas representable by finite-state automata: unsatisfiability of a formula corresponds to emptiness of its automaton, which is testable efficiently in time linear in the automaton size.

It is natural to extend this framework \cite{Berstel79-book,Saka09-book} to represent $n$-ary predicates, for $n > 1$, by means of \emph{multi-tape} finite-state automata.
An $n$-tape automaton $A_R$ is a device that accepts $n$-tuples of words, corresponding to the set of models of a predicate $R(x_1, \ldots, x_n)$ over $n$ variables.
Section~\ref{sec:preliminaries} defines multi-tape automata and summarizes some of their fundamental properties. 
It turns out that the class of multi-tape automata (in their most expressive \emph{asynchronous} variant) is not as robust as one-tape automata.
In particular, multi-tape automata\footnote{We do not consider more powerful  classes of multi-tape automata, such as pushdown automata, as they typically possess even fewer closure or decidability properties~\cite{Ibarra78} unless they are significantly restricted to specific classes of languages~\cite{EsparzaGM12}.} are not closed under intersection~\cite{CAF-survey}, and hence the conjunction of $n$-ary predicates is not implementable in general. % by multi-tape automata.
% This is a hurdle towards the goal of using multi-tape automata to decide the satisfiability of first-order formulas extending what has been done for monadic logic: the conjunction of $n$-ary predicates may not be implementable by multi-tape automata.

This paper investigates the magnitude of this hurdle in practice.
On the negative side, we prove that determining whether the intersection of two multi-tape automata $A, B$ is expressible as an automaton is neither decidable nor semi-decidable.
On the positive side, we provide an algorithm $\algo{A, B, d}$ that computes an under-approximation of the intersection $A \cap B$ of $A$ and $B$, bounded by a given maximum delay $d$ between heads on different tapes.
We also detail a simple sufficient syntactic condition on $A$ and $B$ for the algorithm to return complete intersections.
Based on these results, we implemented the algorithm and tried it on a number of natural examples inspired by the verification conditions of programs operating on sequences.

\section{Preliminaries} \label{sec:preliminaries}

$\integers$ is the set of integer numbers, and $\naturals$ is the set of natural numbers $0, 1, \ldots$.
For a (finite) set $S$, $\powerset{S}$ denotes its powerset.
For a finite nonempty \emph{alphabet} $\Sigma$, $\Sigma^*$ denotes the set of all finite sequences $\sigma_1 \cdots \sigma_n$, with $n \geq 0$, of symbols from $\Sigma$ called \emph{words} over $\Sigma$; when $n = 0$, $\epsilon \in \Sigma^*$ is the \emph{empty} word.
$|s| \in \naturals$ denotes the length $n$ of a word $s = \sigma_1 \,\ldots\, \sigma_n$.
An \emph{$n$-word} is an $n$-tuple $\langle s_1, \ldots, s_n \rangle \in (\Sigma^*)^n$ of words over $\Sigma$. %, where each $s_i = \sigma_{1,i}\,\ldots\,\sigma_{n_i,i}$.

Given a sequence $s = x_1 \cdots x_n$ of objects, a \emph{permutation} $\pi: \{1, \ldots, n\} \to \{1, \ldots, n\}$ is a bijection that rearranges $s$ into $\pi(s) = \pi_1 \cdots \pi_n$ with $\pi_i = x_{\pi(i)}$ for $i = 1, \ldots, n$.
An \emph{inversion} of a permutation $\pi$ of $s$ is a pair $(i, j)$ of indices such that $i < j$ and $\pi(i) > \pi(j)$.
For example, the permutation that turns $a_4b_1b_2a_5a_6a_7b_3$ into $b_1b_2b_3a_4a_5a_6a_7$ has $6$ inversions.

\subsection{Multi-Tape Finite Automata}
A finite-state automaton with $n \geq 1$ tapes scans $n$ read-only input tapes, each with an independent head. At every step, the current state determines the tape to be read, and the transition function defines the possible next states based on the current state and the symbols under the reading head.
A special symbol $\emk$ marks the right end of each input tape; $\Sigmaend$ denotes the extended alphabet $\Sigma \cup \{\emk\}$.
\begin{Definition}[$n$-tape automaton] \label{def:n-tape-aut}
An \emph{$n$-tape finite-state automaton} $A$ is a tuple $\langle \Sigma, Q, \linebreak \delta, Q_0, F, T, \tau\rangle$ where: $\Sigma$ is the input alphabet, with $\emk \not\in \Sigma$; $T = \{t_1, \ldots, t_n\}$ is the set of tapes; $Q$ is the finite set of states; $\tau: Q \to T$ assigns a tape to each state; $\delta: Q \times \Sigmaend \to \powerset{Q}$ is the (nondeterministic) transition function; $Q_0 \subseteq Q$ are the initial states; $F \subseteq Q$ are the accepting (final) states.
\end{Definition}
We write $A(t_1, \ldots, t_n)$ when we want to emphasize that $A$ operates on the $n$ tapes $t_1, \ldots, t_n$; $A(t_1', \ldots, t_n')$ denotes an instance of $A$ with each tape $t_i$ renamed to $t_i'$.
Without loss of generality, assume that the accepting states have no outgoing edges: $\delta(q_F, \sigma) = \emptyset$ for all $q_F \in F$.
Also, whenever convenient we represent the transition function $\delta$ as a relation, that is the set of triples $(q, \sigma, q')$ such that $q' \in \delta(q, \sigma)$.

A \emph{configuration} of an $n$-tape automaton $A$ is an $(n+1)$-tuple $\langle q, y_1, \ldots, y_n \rangle \in Q \times (\Sigmaend^*)^n$, where $q \in Q$ is the current state and, for $1 \leq k \leq n$, $y_k$ is the input on the $k$-th tape still to be read.
A \emph{run} $\rho$ of $A$ on input $x = \langle x_1, \ldots, x_n \rangle \in (\Sigma^{*})^n$ is a sequence of configurations $\rho = \rho_0 \cdots \rho_m$ such that: (1) $\rho_0 = \langle q_0, x_1\emk, \ldots, x_n\emk \rangle$ for some initial state $q_0 \in Q_0$; and (2) for $0 \leq k < m$, if $\rho_k = \langle q, y_1, \ldots, y_n \rangle$ is the $k$-th configuration---with $t_h = \tau(q)$ the tape read in state $q$, and $y_h = \sigma\, y_h'$ with $\sigma \in \Sigmaend$ and $y_h' \in \Sigmaend^*$ on the $h$-th tape---then $\rho_{k+1} = \langle q', y_1', \ldots, y_n'\rangle$ with $q' \in \delta(q, \sigma)$ and $y_i' = y_i$ for all $i \neq h$.
A run $\rho$ is \emph{accepting} if $\rho_m = \langle q_F, y_1, \ldots, y_n \rangle$ for some accepting state $q_F \in Q_F$.
$A$ accepts an $n$-word $x$ if there exists an accepting run of $A$ on $x$.
The \emph{language} accepted (or recognized) by $A$ is the set $\lang{A}$ of all $n$-words that $A$ accepts.
The \emph{$n$-rational languages} are the class of languages accepted by some $n$-tape automaton.
Whenever $n$ is clear from the context, we will simply write ``words'' and ``automata'' to mean ``$n$-words'' and ``$n$-tape automata''.

\begin{Definition}
An $n$-tape automaton $A$ is: 
\textbf{deterministic}
if $|Q_0| \leq 1$ and $|\delta(q, \sigma)| \leq 1$ for all $q, \sigma$;
\textbf{synchronous} for $s \in \naturals$ if every run of $A$ is such that any two heads that have not scanned their whole input are no more than $s$ positions apart;
\textbf{asynchronous} if it is not synchronous for any $s$.
\end{Definition}

\begin{example}\label{ex:eq-cat-automata}
Figure~\ref{fig:Aeq} shows a synchronous deterministic automaton $\aut{fig:Aeq}$ with two tapes $X, Y$ that recognizes pairs of equal words over $\{a,b\}$.
Each state is labeled with the tape read and with a number for identification (the final state's tape label is immaterial, and hence omitted).
$\aut{fig:Aeq}$ reads one letter on tape $Y$ immediately after reading one letter on tape $X$; hence it is synchronous for $s=1$.
Automaton $\aut{fig:Acat}$ in Figure~\ref{fig:Acat} recognizes triples of words such that the word on tape $Z$ equals the concatenation of the words on tapes $X$ and $Y$ (ignoring the end-markers). It is asynchronous because the length of $X$ is not bounded: when the reading on tape $Y$ starts, the head on $Z$ is at a distance equal to the length of the input on $X$.
\end{example}

\begin{figure}[!ht]
\centering
\begin{tikzpicture}[->,semithick, accepting/.style=accepting by double,initial text=,
  every state/.style={draw,minimum size=6mm,inner sep=1pt},
  state/.style=state with output,
  shorten >=2pt, shorten <=2pt, node distance=20mm]
\node [state,initial,initial where=above] (eq1) {$X$ \nodepart{lower} $1$};
\node [state] (eq2) [left=of eq1] {$Y$ \nodepart{lower} $2$};
\node [state] (eq3) [right=of eq1] {$Y$ \nodepart{lower} $3$};
\node [state] (eq4) [right=of eq3] {$Y$ \nodepart{lower} $4$};
\node [state,accepting] (eq5) [right=of eq4] {\nodepart{lower} $5$};

\path (eq1) edge [bend right] node [above] {$a$} (eq2);
\path (eq1) edge [bend left] node [above] {$b$} (eq3);
\path (eq2) edge [bend right] node [above] {$a$} (eq1);
\path (eq3) edge [bend left] node [above] {$b$} (eq1);
\path (eq1.south) edge [bend right=20] node [near end,above] {$\emk$} (eq4);
\path (eq4) edge node [above] {$\emk$} (eq5);
\end{tikzpicture}
\caption{2-tape deterministic synchronous automaton $\aut{fig:Aeq}$.}
\namelabel{fig:Aeq}{=}
\end{figure}
\begin{figure}[!ht]
\centering
\begin{tikzpicture}[->,semithick, accepting/.style=accepting by double,initial text=,
  every state/.style={draw,minimum size=6mm,inner sep=1pt},
  state/.style=state with output,
  shorten >=2pt, shorten <=2pt, node distance=2mm and 10mm]
\node [state,initial,initial where=below] (cat1) {$X$\nodepart{lower} $1$};
\node [state] (cat2) [above left=of cat1] {$Z$\nodepart{lower} $2$};
\node [state] (cat3) [above right=of cat1] {$Z$\nodepart{lower} $3$};
\node [state] (cat4) [right=32mm of cat1] {$Y$\nodepart{lower} $4$};
\node [state] (cat5) [above left=of cat4] {$Z$\nodepart{lower} $5$};
\node [state] (cat6) [above right=of cat4] {$Z$\nodepart{lower} $6$};
\node [state] (cat7) [right=32mm of cat4] {$Z$\nodepart{lower} $7$};
\node [state,accepting] (cat8) [right=12mm of cat7] {\nodepart{lower} $8$};

\path (cat1) edge [bend right] node [above] {$a$} (cat2);
\path (cat1) edge [bend left] node [above] {$b$} (cat3);
\path (cat2) edge [bend right] node [above] {$a$} (cat1);
\path (cat3) edge [bend left] node [above] {$b$} (cat1);

\path (cat4) edge [bend right] node [above] {$a$} (cat5);
\path (cat4) edge [bend left] node [above] {$b$} (cat6);
\path (cat5) edge [bend right] node [above] {$a$} (cat4);
\path (cat6) edge [bend left] node [above] {$b$} (cat4);

\path (cat1) edge [bend right] node [above] {$\emk$} (cat4);
\path (cat4) edge [bend right] node [above] {$\emk$} (cat7);
\path (cat7) edge node [above] {$\emk$} (cat8);
\end{tikzpicture}
\caption{3-tape deterministic asynchronous automaton $\aut{fig:Acat}$.}
\namelabel{fig:Acat}{\cat}
\end{figure}

\subsection{Closure Properties and Decidability} \label{sec:clos-prop-decid}

Automata define languages, which are sets of words; correspondingly, we are interested in the closure properties of automata with respect to set-theoretic operations on their languages.
Specifically, we consider closure under complement, intersection, and union; and the \emph{emptiness} problem: given an automaton $A$, decide whether $\lang{A} = \emptyset$, that is whether it accepts some word.
The complement of a language $L$ over $n$-words over $\Sigma$ is taken with respect to the set $(\Sigma^*)^n$; the intersection $L_1 \cap L_2$ is also applicable when $L_1$ is a language over $n$-words and $L_2$ a language over $m$-words, with $m > n$: define $L_1 \cap L_2$ as the set of $m$-tuples $\langle x_1, \ldots, x_m \rangle$ such that $\langle x_1, \ldots, x_n \rangle \in L_1$ and $\langle x_1, \ldots, x_m \rangle \in L_2$; a similar definition works for unions.
We lift set-theoretic operations from languages to automata; for example, the intersection $A_3 = A_1 \cap A_2$ of two automata $A_1, A_2$ is an automaton $A_3$ such that $\lang{A_3} = \lang{A_1} \cap \lang{A_2}$, when it exists; we assume that intersected automata share the tapes with the same name (in the same order).
The rest of this section summarizes the fundamental closure properties of multi-tape automata; see~\cite{CAF-survey} for a more detailed presentation and references.

\fakepar{Synchronous automata}~\cite{IbarraT12-CIAA,IbarraT12-TCS}
define a very robust class of languages: they have the same expressiveness whether deterministic or nondeterministic; they are closed under complement, intersection, and union; and emptiness is decidable.
In fact, computations of synchronous $n$-tape automata can be regarded as computations of standard single-tape automata over the $n$-track alphabet $(\Sigma \cup \{\pad\})^n$, where the fresh symbol $\pad$ pads some of the $n$ input strings so that they all have the same length.
Under this convention, the standard constructions for finite-state automata apply to synchronous automata as well.
Most applications of multi-tape automata to have targeted synchronous automata (see Section~\ref{sec:related-work}), which have, however, a limited expressive power.

\fakepar{Asynchronous automata}
are strictly more expressive than synchronous ones, but are also less robust:
\begin{itemize}
\item
Nondeterministic asynchronous automata are strictly more expressive than deterministic ones.
\item
\emph{Deterministic} asynchronous automata are closed under complement, using the standard construction that complements the accepting states.
% (after completing the transition function).  
They are not closed under union, although the union of two deterministic asynchronous automata always is a nondeterministic automaton.
They are not closed under intersection because, intuitively, the parallel computations in the two intersected automata may require the heads on the shared tapes to diverge.
\item
\emph{Nondeterministic} asynchronous automata are not closed under complement or intersection, but are closed under union using the standard construction that takes the union of the transition graphs.
\item
\emph{Emptiness} is decidable for asynchronous automata (deterministic and nondeterministic): it amounts to testing reachability of accepting states from initial states on the transition graph.
\end{itemize}

\section{Multi-Tape Automata: Negative Results} \label{sec:negative-results}
Since multi-tape automata are not closed under intersection, we try to characterize the class of intersections that \emph{are} expressible as automata.
A logical characterization is arduous to get, because conjunction would be inexpressible in general.
Indeed, we can prove some strong undecidability results.

\fakepar{Rational intersection is undecidable.}
The \emph{rational intersection problem} is the problem of determining whether the intersection language $\lang{A} \cap \lang{B}$ of two automata $A$ and $B$ is rational, that is whether it is accepted by some multi-tape automaton.

\begin{theorem} \label{th:not-semidecidable}
The  rational intersection problem is not semidecidable.
\end{theorem}
\begin{proof}
Following~\cite{Hart67,abs-0906-3051}, we consider valid computations of Turing machines.
A single-tape Turing machine $M$ has state set $S$, input alphabet $I$, transition relation $\delta \subseteq Q \times I \times Q \times I \times \{-1, 0, 1\}$; and $s_0, s_F \in S$ respectively are the initial state and the accepting state (unique, without loss of generality).
We can write $M$'s configurations as strings over $I \cup S$ of the form $i_1 \cdots i_k \,s\,i_{k+1} \cdots i_m$ where $i_1 \cdots i_m \in I^*$ is the sequence of symbols on the tape, $s$ is the current state, and the read/write head is over the symbol $i_{k+1}$.
The set $\valc{M}$ of accepting computations contains all words of the form $\# w_1 \# \cdots \# w_m \#$, with $\# \not\in I \cup S$, such that each $w_k$ is a configuration of $M$, $w_1$ is an initial configuration (of the form $s_0 I^*$), $w_m$ is an accepting configuration (of the form $I^* s_F I^*$), and $w_{k+1}$ is a valid successor of $w_k$ according to $\delta$, for all $1 \leq k < m$ (that is, for $w_k = i_1 \cdots i^-\, s\, i^+ \cdots i_n$ and $(s, i^+, s', i', h) \in \delta$ then: if $h = -1$ then $w_{k+1} = i_1 \cdots s' i^- i' \cdots i_n$; if $h = 0$ then $w_{k+1} = i_1 \cdots i^- s i' \cdots i_n$; if $h = +1$ then $w_{k+1} = i_1 \cdots i^- i's \cdots i_n$).
The problem of determining, for a generic $M$, whether $\valc{M}$ is regular is not semidecidable~\cite{Hart67}.

Consider now the language $L_{M^2}$ defined as $\{\langle x, x \rangle \mid x \in \valc{M}\}$.
Since the single-component projection of a rational language is always regular~\cite{RabinScott59}, if $L_{M^2}$ is rational then $\valc{M}$ is regular.
We can express $L_{M^2}$ as the intersection of two languages $L_M^1$ and $L_M^2$.
$L_M^1$ is the set of 2-words $\langle \# u_1 \# \cdots \# u_m \#, \# v_1 \# \cdots \# v_m \# \rangle$ such that: $u_1$ is an initial configuration of $M$; $v_m$ is an accepting configuration; for $1 \leq k < m$, $u_k$ is a valid configuration and $v_{k+1}$ is a valid successor of $u_k$.
$L_M^2$ is simply the set of 2-words whose first and second component are equal.
It is not difficult to see that $L_M^1 \cap L_M^2 = L_{M^2}$ and both $L_M^1$ and $L_M^2$ are rational (and deterministic). 
An automaton for $L_M^2$ works synchronously by alternately reading and comparing one character from each tape, generalizing the automaton in Figure~\ref{fig:Aeq}.
An automaton for $L_M^1$ starts with the second head moving forward to $v_2$; it then compares each $u_k$ and $v_{k+1}$ one character at a time, checking that they are consistent with $M$'s $\delta$.

We can finally prove the theorem by contradiction: assume the rational intersection problem is semidecidable.
Then, the following is a semi-decision procedure for the problem of determining whether $\valc{M}$ is regular.
Construct the automata for $L_M^1$ and $L_M^2$.
If $L_M^1 \cap L_M^2 = L_{M^2}$ is rational, then the semi-decision procedure for rational intersection halts with positive outcome; then we conclude that $\valc{M}$ is regular; otherwise loop forever.
Since regularity of $\valc{M}$ is not semidecidable, we have a contradiction.
\qedhere
\end{proof}

\begin{remark} \label{rem:undec-degree}
The rational intersection problem belongs to $\Sigma_2^0$ in the arithmetical hierarchy.
Consider an enumeration $C_1, C_2, \ldots$ of multi-tape automata.
The rational intersection problem for $A$ and $B$ is expressible as: $\exists z \forall x: C_z(x) \liff A(x) \land B(x)$, where $C_k(x)$ is a predicate that holds iff automaton $C_k$ accepts input $x$.
The formula $P(x, z) = C_z(x) \liff A(x) \land B(x)$ within quantifier scope is recursive (just simulate the automata runs); hence $\forall x: P(x, z)$ is $\Pi_1^0$ and $\exists z \forall x: P(x, z)$ is $\Sigma_2^0$.
\end{remark}

\fakepar{Rational nondeterministic complement is undecidable.}
Recall that deterministic automata are closed under complement and nondeterministic automata are closed under union.
Therefore, the undecidability Theorem~\ref{th:not-semidecidable} carries over to the rational complement problem (defined as obvious): the languages $L_M^1$ and $L_M^2$ used in the proof of Theorem~\ref{th:not-semidecidable} are deterministic; hence their complements $\comp{L_M^1}$ and $\comp{L_M^2}$ are rational languages whose union $\comp{L_M^1} \cup \comp{L_M^2}$ is also rational.
Thus, the complement of $\comp{L_M^1} \cup \comp{L_M^2} = \comp{L_M^1 \cap L_M^2}$ is rational iff $\valc{M}$ is regular.
\begin{corollary}
The rational complement problem (determining whether the complement of a rational nondeterministic language is rational) is not semidecidable.
\end{corollary}

\fakepar{Closure with respect to equality implies synchrony.}
The proof of Theorem~\ref{th:not-semidecidable} also reveals that even the intersection of an asynchronous automaton with a synchronous one is in general not rational.
A natural question is then whether there exist interesting combinations of synchronous and asynchronous automata whose intersection is rational.
A particularly significant case is equality of tapes: it is a relation clearly recognized by synchronous automata, and it plays an important role in the combination of decision procedures (e.g., \`a la Nelson-Oppen~\cite{NelsonO79} or following the DPLL($T$) paradigm~\cite{NieuwenhuisOT06}).
Unfortunately, it is a corollary of standard results that the only ``natural'' and robust class of rational languages are those accepted by synchronous automata.

\begin{corollary} \label{th:asynch-and-equality}
Consider an $n$-tape automaton $A$; if the language $\lang{A} \cap \{ x^n \mid x \in \Sigma^*\}$ is rational, then it is also accepted by a synchronous automaton.
\end{corollary}
\begin{proof}
Language $L = \lang{A} \cap \{ x^n \mid x \in \Sigma^*\}$ is so-called ``length-preserving'': in any word in $L$, all components have the same length.
Theorem~6.1 in~\cite[Chap.~IX]{eilenberg-book} shows that length-preserving rational languages are synchronous.
\qedhere
\end{proof}

\section{Multi-Tape Automata: Positive Results} \label{sec:positive-results}
The undecidability of whether an intersection is rational does not prevent the definition of \emph{approximate} algorithms for intersection.
Section~\ref{sec:algor-inters} describes one such algorithm that bounds the maximum delay between corresponding heads of the intersecting automata; the algorithm under-approximates  the real intersection.
We also discuss very simple syntactic conditions under which a bound of zero delay still yields a complete intersection.
Section~\ref{sec:approx-complement} discusses to what extent some of these results can be extended to the approximation of complement for nondeterministic automata.

\subsection{An Algorithm for the Under-Approximation of Intersection} \label{sec:algor-inters}
This section outlines an algorithm $\algo{A, B, d}$ that inputs two multi-tape automata $A$ and $B$ and a delay bound $d \in \naturals \cup \{\infty\}$ and returns a multi-tape automaton $C$ that approximates the intersection $A \cap B$ to within delay $d$.
The intersection construction extends the classic ``cross-product'' construction: simulate the parallel runs of the two composing automata by keeping track of what happens in each component.

\fakepar{Informal overview.}
Let us introduce the algorithm's basics through examples.
Consider the intersection of $\aut{fig:Aeq}$ and $\aut{fig:Acat}$ in Figures~\ref{fig:Aeq} and \ref{fig:Acat}; the initial state is labeled $\langle \mathop{=}_1, \mathop{\cat}_1 \rangle$ to denote that it combines states $\mathop{=}_1$ (i.e., state $1$ in $\aut{fig:Aeq}$) and $\mathop{\cat}_1$ (state $1$ in $\aut{fig:Acat}$).
As the intersection develops, the composing automata synchronize on transitions on shared tapes and proceed asynchronously on non-shared tapes.
In the example, there is a synchronized transition from $\langle \mathop{=}_1, \mathop{\cat}_1 \rangle$ to $\langle \mathop{=}_2, \mathop{\cat}_2 \rangle$ upon reading $a$ on shared tape $X$, and an asynchronous transition from the latter state to $\langle \mathop{=}_2, \mathop{\cat}_1 \rangle$ upon reading $a$ on $\aut{fig:Acat}$'s non-shared tape $Z$.
$\aut{fig:Aeq}$ in state $\mathop{=}_2$ can also read $a$ on shared tape $Y$; this is a valid move in the intersection even if $\aut{fig:Acat}$ cannot read on tape $Y$ until it reaches state $\mathop{\cat}_4$. Since reading can proceed on other tapes, we just have to ``delay'' the transition that reads $a$ on $Y$ to a later point in the computation and store this delay using the states of the intersection automaton; $\aut{fig:Acat}$ will then be able to take other transitions and will consume the delayed ones asynchronously \emph{before} taking any other transition on $Y$ (that is, delays behave as a FIFO queue).
For example, when a run of the intersection automaton reaches state $\langle \mathop{=}_4, \mathop{\cat}_4 \rangle$, $\aut{fig:Acat}$ can read $a$ on $Y$ matching $\aut{fig:Aeq}$'s delayed transition (which is then consumed).
Here is a picture showing these steps:
 
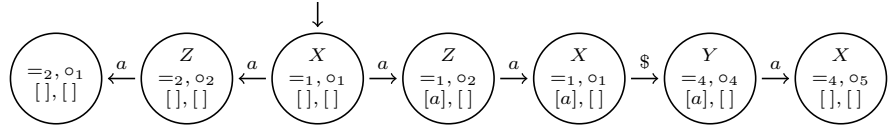
\begin{tikzpicture}[->,semithick, accepting/.style=accepting by double,initial text=,
  every state/.style={draw,minimum size=12mm,inner sep=1pt},align=center,
  shorten >=2pt, shorten <=2pt, node distance=2mm and 5mm, font=\scriptsize]

\node [state,initial,initial where=above] (s11x) {$X$ \\ $\mathop{=}_1, \mathop{\cat}_1$ \\ $[\ ], [\ ]$};
\node [state] (s22z) [left=of s11x] {$Z$ \\ $\mathop{=}_2, \mathop{\cat}_2$ \\ $[\ ], [\ ]$};
\node [state] (s21x) [left=of s22z] {$\ $ \\ $\mathop{=}_2, \mathop{\cat}_1$ \\ $[\ ], [\ ]$};
\node [state] (s12z) [right=of s11x] {$Z$ \\ $\mathop{=}_1, \mathop{\cat}_2$ \\ $[a], [\ ]$};
\node [state] (s11xd) [right=of s12z] {$X$ \\ $\mathop{=}_1, \mathop{\cat}_1$ \\ $[a], [\ ]$};
\node [state] (s2) [right=of s11xd] {$Y$ \\ $\mathop{=}_4, \mathop{\cat}_4$ \\ $[a], [\ ]$};
\node [state] (s3) [right=of s2] {$X$ \\ $\mathop{=}_4, \mathop{\cat}_5$ \\ $[\ ], [\ ]$};

\path (s11x) edge  node [above] {$a$} (s22z);
\path (s11x) edge node [above] {$a$} (s12z);
\path (s12z) edge node [above] {$a$} (s11xd);
\path (s22z) edge node [above] {$a$} (s21x);

\path (s11xd) edge node [above] {$\emk$} (s2);
\path (s2) edge node [above] {$a$} (s3);
\end{tikzpicture}

Delays may become unbounded in some cases.
In the example, automaton $\aut{fig:Aeq}$ may accumulate arbitrary delays on tape $Y$ while in states $\mathop{=}_1, \mathop{=}_2, \mathop{=}_3$; this corresponds to the intersection automaton ``remembering'' an arbitrary word on tape $Y$ to compare it against $Z$'s content later.
An unbounded delay is necessary in this case, as the computations on $\aut{fig:Aeq}$ and $\aut{fig:Acat}$ manage the heads on $X$ and $Y$ in irreconcilable ways: the intersection language of $\aut{fig:Aeq}$ and $\aut{fig:Acat}$ is not rational.

\fakepar{The algorithm.}
Consider two automata $A = \langle \Sigma, Q^A, \delta^A, Q^A_0, F^A, T^A, \tau^A \rangle$ and $B = \langle \Sigma, Q^B, \delta^B, Q^B_0, F^B, T^B, \tau^B \rangle$, such that $A$ has $m$ tapes $T^A = \{t_1^A, \ldots,  t_m^A\}$ and $B$ has $n$ tapes $T^B = \{t_1^B, \ldots, t_n^B\}$.
We present an algorithm $\algo{A, B, d}$ that constructs an automaton $C = \langle \Sigma, Q, \delta, Q_0, F, T, \tau \rangle$---with $C$'s tapes $T = T^A \cup T^B$---such that $\lang{C} \subseteq \lang{A} \cap \lang{B}$.
We describe the algorithm as the combination of fundamental operations, introduced as separate routines.
All components of the algorithm have access to the definitions of $A$ and $B$, to the definition of $C$ being built, and to a global stack \lstinline|s| where new states of the composition are pushed (when created) and popped (when processed).
The complete pseudo-code of the routines is \seeAppendix{app:under-appr-inters}.

Routine \lstinline|async_next| (lines~1--17 in Figure~\ref{fig:asyncNext}) takes a $t$-tape automaton $D$ (i.e., $A$ or $B$) and one of its states $q$, and returns a set of tuples $\langle q', h_1, \ldots, h_t \rangle$ of all next states reachable from $q$ by accumulating delayed transitions $h_i \in (\delta^D)^*$ in tape $t_i$, for $1 \leq i \leq t$.
We call \emph{delayed states} such tuples of states with delayed transitions.
The search for states reachable from $q$ stops at the first occurrences of states associated with a certain tape.
For example, \lstinline|async_next($\aut{fig:Acat}$, $\mathop{\cat}_1$)| consists of $\langle \mathop{\cat}_1, \epsilon, \epsilon, \epsilon \rangle$, $\langle \mathop{\cat}_2, (\mathop{\cat}_1, a, \mathop{\cat}_2), \epsilon, \epsilon \rangle$, $\langle \mathop{\cat}_3, (\mathop{\cat}_1, b, \mathop{\cat}_3), \epsilon, \epsilon \rangle$,  $\langle \mathop{\cat}_4, (\mathop{\cat}_1, \emk, \mathop{\cat}_4), \epsilon, \epsilon \rangle$, $\langle \mathop{\cat}_5, (\mathop{\cat}_1, \emk, \mathop{\cat}_4), (\mathop{\cat}_4, a, \linebreak \mathop{\cat}_5), \epsilon \rangle$, $\langle \mathop{\cat}_6, (\mathop{\cat}_1, \emk, \mathop{\cat}_4), (\mathop{\cat}_4, b, \mathop{\cat}_6), \epsilon \rangle$, and $\langle \mathop{\cat}_7, (\mathop{\cat}_1, \emk, \mathop{\cat}_4), (\mathop{\cat}_4, \emk, \mathop{\cat}_7), \epsilon \rangle$.

Consider now a pair of delayed states $\langle p, h_1, \ldots, h_m \rangle$ and $\langle q, k_1, \ldots, k_n \rangle$, respectively of $A$ and $B$.
The two states can be composed only if the delays on the synchronized tapes are pairwise \emph{consistent}, that is the sequence of input symbols of one is a prefix (proper or not) of the other's; otherwise, the intersection will not be able to consume the delays in the two components because they do not match.
\lstinline|cons($h_i, k_i$)| denotes that the sequences $h_i, k_i$ of delayed transitions are consistent.
Routine \lstinline|new_states| (lines~19--26 in Figure~\ref{fig:newStates}) takes two sets $P, Q$ of delayed states and returns all consistent states obtained by composing them.
\lstinline|new_states| also pushes onto the stack \lstinline|s| all composite states that have not already been added to the composition.
For convenience, \lstinline|new_state| also embeds the tape $t$ of each new composite state within the state itself.
(All tapes are considered: states corresponding to inconsistent choices will be dead ends.)

To add arbitrary prefixes to the delays of delayed states generated by \lstinline|new_states|, routine \lstinline|compose_transition| (lines~28--33 in Figure~\ref{fig:composeTransition}) takes two sets $P, Q$ of delayed states and an $(m+n)$-tuple of delays, and calls \lstinline|new_states| on the modified states obtained by orderly adding the delays to the states in $P$ and $Q$.
It also adds all transitions reaching the newly generated states to $C$'s transition function $\delta$.

We are ready to describe the main routine \lstinline|intersect| which builds $C$ from $A$ and $B$; see Figure~\ref{fig:intersect} for the pseudo-code (some symmetric cases are omitted for brevity).
Routine \lstinline|intersect| takes as arguments a bound on the maximum number of states and on the maximum delay \lstinline|max_delay| (measured in number of transitions) accumulated in the states.
After building the initial states of the compound (lines 4--5), \lstinline|intersect| enters a loop until either no more states are generated (i.e., the stack \lstinline|s| is empty) or it has reached the bound \lstinline|max_states| on the number of states.
Each iteration of the loop begins by popping a state $r$ from the top of the stack (line 7).
$r$ is normally added to the set $Q$ of $C$'s states, unless some of its sequences of delayed transitions are longer than the delay bound \lstinline|max_delay|; in this case, the algorithm discards $r$ and proceeds to the next iteration of the loop (line 8).
If $r$ is not discarded, \lstinline|intersect| builds all composite states reachable from $r$.
These depend on the tape $t$ read when in $r$: if it is shared between $A$ and $B$ we have synchronized transitions (lines 10--30), otherwise we have an asynchronous transition of $A$ (lines 32--41) or one of $B$ (line 43).

Consider the case of a synchronized transition on some shared tape $t \in T^A  \cap T^B$.
While both $A$ and $B$ must read the same symbol on the same tape, they may do so by consuming some transition that has been delayed.
For example, if $A$ has a non-empty delay $h_t \neq \epsilon$ for tape $t$, it will consume the first transition $(u_a, \sigma, u_a')$ in $h_t$; since the transition is delayed, $A$'s next state in the compound is not determined by the delayed transition (which only reads the input $\sigma$ at a delayed instant) but by $A$'s current state $q_a$ in the compound (line 12 and line 17).
The reached states are the composition of those reached within $A$ and $B$, with the delays updated so as to remove the delayed transitions consumed.
For example, lines 12--14 correspond to both $A$ and $B$ taking a delayed transition, whereas lines 17--20 correspond to $A$ taking a delayed transition and $B$ taking a ``normal'' transition determined by its transition function $\delta^B$ on symbol $\sigma$.
If neither $A$ nor $B$ have delayed transitions for tape $t$, they can only perform normal transitions according to their transitions functions, without consuming the delays stored in the state; this is shown in lines 26--30.

The final portion of \lstinline|intersect| (from line 32) handles the case of transitions on some non-shared tape $t$.
In these cases, the component of the state $r$ corresponding to the automaton that does not have tape $t$ does not change at all, whereas the other component is updated as usual---either by taking a delayed transition (lines 33--35) or by following its transition function (lines 37--41).

The output of $\algo{A, B, d}$ coincides with the main routine \mbox{\lstinline|intersect|} called on $A$ and $B$ with no bound on the number of states and \lstinline|max_delay = $d$|; the final states $F$ in $C$ coincide with those whose components are both final in $A$ and $B$ and have no delayed transitions.

\subsection{Correctness and Completeness} \label{sec:correct-precise}
In the proofs of this section, we make the simplifying assumption that all tapes are shared: $T^A = T^B = T$; handling non-shared tapes is straightforward.
Let us show that $\algo{A, B, d}$ is correct, that is it constructs an under-approximation of the intersection.

\begin{theorem}[Correctness] \label{th:correct}
For every finite delay $d \in \naturals$, $\algo{A, B, d}$ returns a $C$ such that $\lang{C} \subseteq \lang{A} \cap \lang{B}$.
\end{theorem}
\begin{proof}
Let us show that $x \in \lang{C}$ implies $x \in \lang{A} \cap \lang{B}$.
The basic idea is that, given an accepting run $\rho = \rho_0 \rho_1 \cdots \rho_n$ of $C$ on $x$, one can construct two permutations $\pi^A, \pi^B$ such that $\pi^A(\rho)$ is an accepting run of $A$ and $\pi^B(\rho)$ is an accepting run of $B$ on $x$.
The permutation $\pi^A$ is constructed as follows (constructing $\pi^B$ works in the same way): each element $\rho_k$ in $\rho$, for $0 \leq k < n$, corresponds to either a synchronous or a delayed transition of $A$; in the former case, $\pi^A$ does not change the position of $\rho_k$, otherwise it moves it to where the transition was delayed (i.e., consumed asynchronously).
For accepting runs, it is always possible to construct such permutations, since accepting states in $C$ have no delays, and hence delayed transitions must have been consumed somewhere \emph{before} reaching the accepting state.
\qedhere
\end{proof}

\begin{remark}[Termination] \label{rem:stopping}
$\algo{A, B, d}$ only expands states encoding a maximum delay of $d$, and terminates after the given maximum number of states have been generated or when all states have been explored---whatever comes first.
Upon termination, in general we do not know if the generated intersection automaton accepts $\lang{C}$ or only a subset of it---consistently with the undecidability degree of deciding the intersection (Remark~\ref{rem:undec-degree}).
\end{remark}

\begin{remark}[Complexity] \label{rem:complexity}
Since $\algo{A, B, d}$ may have to enumerate all $d$-delayed states, its worst-case time complexity is exponential in $d$ and $|T|$---which determine the combinatorial explosion in the number of compound states---as we now illustrate.

To get an upper bound on the time complexity of $\algo{A, B, d}$, let $q = \max(|Q^A|, \linebreak |Q^B|)$, $\delta = \max(|\delta^A|, |\delta^B|) = \Orm(q^2)$, and $t = |T^A \cap T^B| = |T^A| = |T^B|$.
Consider the case in which \lstinline|intersect| expands \emph{all} $d$-delay compound states determined by $A$ and $B$, and ignore constant multiplicative factors.
The main loops executes once per compound state, that is $q^2 \delta^{2 d t}$ times.
Each iteration: (1) calls \lstinline|async_next| on $A$ and $B$, taking time $t q^{3}$ using an algorithm such as Floyd-Warshall for the all-pairs shortest path (but whose results can be cached); and (2) composes the sets of states by calling \lstinline|compose_transition|, taking time $d t q^4$ assuming states in $Q$ are hashed.
The dominant time-complexity factor is then $d t q^6 q^{4dt}$, exponential in $d$ and $t$.
\end{remark}

There is a simple condition for $\algo{A, B, d}$ to return complete intersections.
If $A$ and $B$ share only one tape, $\algo{A, B, 0}$ returns a $C$ that reads the input on non-shared tapes asynchronously whenever possible; otherwise, $C$ reads synchronously the input on the single shared tape without need to accumulate delays.
\begin{lemma}[Sufficient condition for completeness] \label{lem:completeness}
If $A$ and $B$ share at most one tape, then $\algo{A, B, 0}$ returns a  $C$ such that $\lang{C} = \lang{A} \cap \lang{B}$.
\end{lemma}

\begin{example}
Consider the intersection of $A_1 = \aut{fig:Acat}(X,Y,Z)$ and $A_2 = \aut{fig:Aeq}(Z, W)$ (the latter is $\aut{fig:Aeq}$ in Figure~\ref{fig:Acat} with tapes renamed to $Z$ and $W$).
Since $A_1$ and $A_2$ only share tape $Z$, they can be ready to read synchronously on $Z$ whenever necessary without having to delay such transitions, since asynchronous transitions can be interleaved ad lib.
Therefore, bounding the construction to have no delays gives an automaton that accepts precisely the intersection of $A_1$'s and $A_2$'s languages.
\end{example}

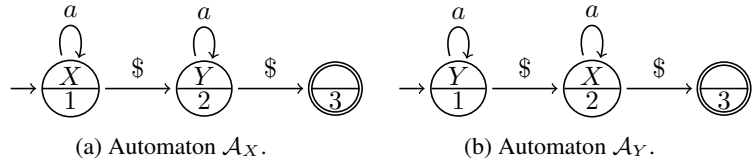
\begin{figure}[!ht]
\centering
\begin{subfigure}[b]{0.4\textwidth}
\begin{tikzpicture}[->,semithick, accepting/.style=accepting by double,initial text=,
  every state/.style={draw,minimum size=6mm,inner sep=1pt},
  state/.style=state with output,
  shorten >=2pt, shorten <=2pt, node distance=10mm]
\node [state,initial,initial where=left] (A1) {$X$ \nodepart{lower} $1$};
\node [state] (A2) [right=of A1] {$Y$ \nodepart{lower} $2$};
\node [state,accepting] (A3) [right=of A2] {\nodepart{lower} $3$};

\path (A1) edge [loop above] node [above] {$a$} (A1);
\path (A1) edge node [above] {$\emk$} (A2);
\path (A2) edge [loop above] node [above] {$a$} (A2);
\path (A2) edge node [above] {$\emk$} (A3);
\end{tikzpicture}
\caption{Automaton $\aut{fig:anyXfirst}$.}
\namelabel{fig:anyXfirst}{X}
\end{subfigure}
~
\begin{subfigure}[b]{0.4\textwidth}
\begin{tikzpicture}[->,semithick, accepting/.style=accepting by double,initial text=,
  every state/.style={draw,minimum size=6mm,inner sep=1pt},
  state/.style=state with output,
  shorten >=2pt, shorten <=2pt, node distance=10mm]
\node [state,initial,initial where=left] (A1) {$Y$ \nodepart{lower} $1$};
\node [state] (A2) [right=of A1] {$X$ \nodepart{lower} $2$};
\node [state,accepting] (A3) [right=of A2] {\nodepart{lower} $3$};

\path (A1) edge [loop above] node [above] {$a$} (A1);
\path (A1) edge node [above] {$\emk$} (A2);
\path (A2) edge [loop above] node [above] {$a$} (A2);
\path (A2) edge node [above] {$\emk$} (A3);
\end{tikzpicture}
\caption{Automaton $\aut{fig:anyYfirst}$.}
\namelabel{fig:anyYfirst}{Y}
\end{subfigure}
\caption{Two automata accepting the same language $\{\langle a^m, a^n\rangle \mid m, n \in \naturals\}$.}
\label{fig:anyXandY}
\end{figure}

\begin{remark}
Even when called without bound on the delays, $\algo{A, B, \infty}$ may terminate; in this case, there is not guarantee on the completeness of the returned $C$.
Consider, for example,\footnote{Thanks to the anonymous Reviewer~1 of CSR 2014 for suggesting this example.} automata $\aut{fig:anyXfirst}$ and $\aut{fig:anyYfirst}$ in Figure~\ref{fig:anyXandY}: they both accept the language $L = \{\langle a^m, a^n\rangle \mid m, n \in \naturals\}$ but by reading on the tapes in different order.
$\algo{\aut{fig:anyXfirst}, \aut{fig:anyYfirst}, \infty}$ terminates and returns a $C$ accepting the language 
\[
L_C = \{\langle a^n, \epsilon\rangle, \langle \epsilon, a^n \rangle \mid n \in \naturals\} \,.
\]
Clearly, $L_C \subset \lang{\aut{fig:anyXfirst}} \cap \lang{\aut{fig:anyYfirst}} = L$.
\end{remark}

\subsection{Approximating Complement} \label{sec:approx-complement}
Since deterministic automata are closed under complement, we can use a construction to approximate determinization to build approximate complement automata.
A straightforward under-approximation algorithm for determinization works as follows.
Consider a generic nondeterministic automaton $A$, and let $b$ be a bound on delays; $\widetilde{A}$ is the approximate deterministic version of $A$ which we construct.
Whenever $A$ has a nondeterministic choice between going from state $q$ to states $q_1$ or $q_2$ upon reading some $\sigma$, $\widetilde{A}$ goes to $q_1$ and continues the computation corresponding to that choice for up to $b$ steps; while performing these $b$ steps, $\widetilde{A}$ stores the symbols read in its finite memory.
If the computation terminates with acceptance within $b$ steps, then $\widetilde{A}$ accepts; otherwise, it continues with the computation that chose to go to $q_2$, using the stored finite input for $b$ steps and then continuing as normal.
It is clear that if such an automaton $\widetilde{A}$ accepts, $A$ accepts as well; the converse is in general not true.
Since $\widetilde{A}$ is deterministic, it can be complemented with the usual construction that switches accepting and non-accepting states.

\begin{remark}
Note that, while deterministic automata are closed under complement, the converse is not true: there exist rational languages whose complement is also rational that are strictly nondeterministic.
For example, consider $L = \{ \langle a^x, a^y \rangle \mid x \neq y \text{ or } x \neq 2y \}$.
It is clear that $L$ is rational; it also requires nondeterminism to ``guess'' whether to check $x \neq y$ (pair each $a$ on the first tape with one $a$ on the second tape) or $x \neq 2y$ (pair each $a$ on the first tape with two $a$'s on the second tape).
$L$'s complement $\overline{L}$ is the singleton set with $\langle \epsilon, \epsilon\rangle$, and hence also rational.
\end{remark}

\section{Implementation and Experiments} \label{sec:impl-exper}
To demonstrate the constructions for multi-tape automata in practice, we implemented the algorithm of Section~\ref{sec:algor-inters} in Python with the IGraph library to represent automata transition graphs; the prototype implementation is about 900 lines long, and includes other basic operations on asynchronous automata such as union, complement (for deterministic), and emptiness test.
Using this prototype, we constructed eight composite automata corresponding to language-theoretic examples and simple verification conditions expressible as the composition of rational predicates, and tested them for emptiness. Table~\ref{tab:experiments} lists the results of the experiments; the examples themselves are described below, and all the formulas are listed \seeAppendix{app:impl-exper-sect}.
All the experiments ran on a Ubuntu GNU/Linux box with Intel Quad Core2 CPU at 2.40~GHz, 4~GB of RAM, Python 2.7.3, and IGraph 0.6.
Each experiment consists of two parts: computing the intersection until (possibly bounded) termination (\textsc{intersection}) and testing the emptiness on the simplified intersection (\textsc{emptiness}).
For each part of each experiment, Table~\ref{tab:experiments} reports the time taken to complete it ($t$, in seconds); for the first part, it also shows the number of states $|Q|$ and transitions $|\delta|$ of the generated automaton; the \textsc{emptiness} column also shows the outcome (?: Y~for empty, N for non-empty), which is, of course, checked to be correct.
Note that the prototype is only a proof-of-concept: there is plenty of room for performance improvement.

\begin{table}
\centering
\setlength{\tabcolsep}{8pt}
\begin{tabular}{l| r rr| r c}
  &  \multicolumn{3}{c|}{\textsc{intersection}}  &  \multicolumn{2}{c}{\textsc{emptiness}} \\
  &  \multicolumn{1}{c}{$t$} & $|Q|$ & $|\delta|$  &  \multicolumn{1}{c}{$t$} & ? \\
\hline
$L_{1,2}$  &  0 & 14 & 6 & 0 & N \\
$L_{3,4}$  &  0 & 56 & 48 & 0 & N \\
\hline
\lstinline|tail|: $\cvc_0$ & 0 & 32 & 32 & 0 & Y \\
\lstinline|tail|: $\cvc_1$ & 0 & 248 & 387 & 0 & Y \\
\lstinline|tail|: $\cvc_2$ & 119 & 1907 & 11061 & 4 & Y  \\
\hline
\lstinline|tail|: $\ce_1$ & 0 & 224 & 564 & 0 & N \\
\lstinline|tail|: $\ce_2$ & 222 & 1644 & 21048 & 28 & N  \\
\hline
$\mathsf{cat}_0$ & 2 & 595 & 1009 & 0 & N \\
\end{tabular}
%\vspace{2mm}
\setlength{\tabcolsep}{6pt}
\caption{Checking languages and verification conditions with multi-tape automata.}
\label{tab:experiments}
\end{table}

\fakepar{Language-theoretic examples.}
Examples $L_{1,2}$ and $L_{3,4}$ (taken from~\cite{KGN04}) 
%A.~Kempe, F.~Guingne, and F.~Nicart. Algorithms for weighted multi-tape automata. Technical Report 031, XRCE, 2004
are 2-word languages whose intersection is finite.
The structure of the automata recognizing the intersected components is such that the algorithm \lstinline|intersect| can only unroll their loops finitely many times, hence terminates without a given bound.
$L_{1,2}$ is the intersection $L_{1,2} = L_1 \cap L_2 = \langle abcabc, abcabca\rangle$ of
$
L_1 =  \left\{\langle ab (cab)^nc, a (bc)^n abca \rangle \mid n \in \naturals\right\}
$ and %\,,
%\quad
$L_2 =  \{\langle(abc)^n, a(bca)^n \rangle \mid n \in \naturals \}$
%\,.
.
$L_{3,4}$ is the intersection $L_{3,4} = L_3 \cap L_4 = \langle a b, xyz \rangle$ of
$
L_3 =  \left\{\langle a b^n, x y^nz  \rangle \mid n \in \naturals\right\}
$
%\,,
%\quad
and
$
L_4 =  \left\{\langle a^n b, x y^nz \rangle \mid n \in \naturals\right\}
$
%\,.
.
It is trivial to build the automata for $L_1, L_2, L_3, L_4$; the experiments reported in Table~\ref{tab:experiments} composed them and determined their finite intersection languages, which happens to be complete for $L_{1,2}$ and $L_{3,4}$.

\fakepar{Program verification examples.}
Consider a routine \lstinline|tail| that takes a nonnegative integer $n$ and a sequence $x$ and returns the sequence obtained by dropping the first $n$ elements of $x$ (where \lstinline|rest(x)| returns \lstinline|x| without its first element):

\begin{footnotesize}
\begin{lstlisting}[numbers=none]
  tail (n: $\naturals$, x: SEQUENCE): SEQUENCE is
     if n = 0 or x = $\epsilon$ then Result := x else Result := tail (n-1, rest(x)) end
\end{lstlisting}
\end{footnotesize}

\noindent
If $|y|$ denotes the length of $y$, a (partial) postcondition for \lstinline|tail| is:
\begin{equation}
(n = 0 \land \Result = x) \lor (n > 0 \land |x| \geq n \land  |\Result| = |x| - n)\,.
\label{eq:post}
\end{equation}

\noindent
The bulk of proving \lstinline|tail| against this specification is showing that the postcondition established by the recursive call in the \lstinline|else| branch (assumed by inductive hypothesis) implies the postcondition \eqref{eq:post}.
Discharging this verification condition is equivalent to proving that three simpler implications, denoted $\vc_0$, $\vc_1$, and $\vc_2$, are valid.
For example:
$
\vc_1 \equiv |y| \geq m \land y = \rest(x) \,\limpl\, |x| \geq n \land m = n - 1
$
states that if sequence $\rest(x)$ has length~$\geq n-1$, then the sequence $x$ has length $\geq n$.

We discharged the verification conditions $\vc_0, \vc_1, \vc_2$ using multi-tape automata constructions as follows.
$\vc_k$ is valid if and only if $\cvc_k = \neg \vc_k$ is unsatisfiable. Hence, we have:
\[
\cvc_1 = \neg \vc_1 \ \equiv\ 
|y| \geq m \land y = \rest(x) \land \left( |x| < n \lor m \neq n - 1 \right)\,.
\]

Assume that sequence elements are encoded with a binary alphabet $\{a,b\}$ and elements of the sequence are separated by a symbol $\#$; this is without loss of generality as a binary alphabet can succinctly encode arbitrary sequence elements.

Then, define multi-tape automata that implement the atomic predicates appearing in the formulas; in all cases, these are very simple and small \emph{deterministic} automata.
For example, define 3 automata $\autd{\len}(X, N)$, $\autd{\rest}(X, Y)$, $\autd{\dec}(M, N)$ for $\cvc_1$.
In $\autd{\len}(X, N)$, tape $X$ stores arbitrary sequences encoded as described above, and tape $N$ encodes a nonnegative integer in unary form (as many $a$'s as the integer); $\autd{\len}(X, N)$ accepts on $X$ sequences whose length (i.e., number of $\#$'s) is no smaller than the number encoded on $N$.
$\autd{\rest}(X,Y)$ accepts if the sequence on tape $Y$ equals the sequence on tape $X$ with the first element (until the first $\#$) removed.
$\autd{\dec}(M,N)$ inputs two nonnegative integers encoded in unary on its tapes $M, N$ and accepts iff $M$ has exactly one less $a$ than $N$.

Finally, compose an overall automaton according to the propositional structure of the formula $\cvc_k$ (using intersection, union, and complement as described \seeAppendix{sec:asynchr-autom-theor}) that is equivalent to it, and test if for emptiness.
For example, $\autd{\cvc_1}$ is equivalent to $\cvc_1$:
\begin{equation}
\autd{\cvc_1} \quad\equiv\quad
\left( \autd{\len}(Y,M) \cap \autd{\rest}(X,Y) \right)
  \cap
\left(
  \comp{\autd{\len}(X,N)}
  \cup
  \comp{\autd{\dec}(M, N)}
\right)\,,
\label{eq:cvc1}
\end{equation}
where $\autd{\len}(Y,M)$ denotes an instance of $\autd{\len}$ with tapes $X, N$ renamed to $Y, M$.
In all cases $\cvc_0, \cvc_1, \cvc_2$, the overall automaton is effectively constructible from the basic automata and each intersection shares only one tape; hence constructing intersections with a zero bound on delays is complete (Lemma~\ref{lem:completeness}).
For example, $\autd{\cvc_1}$ build with zero delays is complete, because each element of the disjunction \eqref{eq:post} is treated separately, as every run of the disjunction automaton is either in $\comp{\autd{\len}(X,N)}$ (that only shares $X$) or in $\comp{\autd{\dec}(M,N)}$ (that only shares $M$).
%Thus, constructing $\autd{\cvc_1}$ with a zero bound on delays in the intersection is complete for $\cvc_1$.

Table~\ref{tab:experiments} shows the results of discharging the verification conditions through this process.
The most complex case is $\cvc_2$ which is the largest formula with 8 variables.
The complete set of verification conditions is shown \seeAppendix{app:impl-exper-sect}.

\paragraph{Failing verification conditions.}
Automata-based validity checking can also detect \emph{invalid} verification conditions by showing concrete counter-examples (assignments of values to variables that make the condition false).
Formulas $\ce_1$ and $\ce_2$ are invalid verification conditions obtained by dropping disjuncts or not complementing them in $\cvc_1$ and $\cvc_2$.
Table~\ref{tab:experiments} shows that the experiments correctly reported non-emptiness.

Even in the cases where the complete intersection is infinite, rational constructions may still be useful to search on-the-fly for accepting states, with the algorithm stopping as soon as it has established that the intersection is not empty.
We did a small experiment in this line with formula $\mathsf{cat}_0$, asserting an incorrect property of sequence concatenation: $x \cat y = z \land \mathit{last}(z) = u \land \mathit{last}(y) = v \limpl u = v$, which does not hold if $y$ is the empty sequence.
Building the intersection with zero delays is not guaranteed to be complete because antecedent and consequent share two variables $u$, $v$; however, it is sufficient to find a counter-example where $y$ is the empty sequence (see Table~\ref{tab:experiments}).

\section{Related Work} \label{sec:related-work}
The study of multi-tape automata began with the classic work of Rabin and Scott~\cite{RabinScott59}.
In the 1960's, Rosenberg and others contributed to the characterization of these automata~\cite{FischerR68,ntape-sync}.
Recent research has targeted a few open issues, such as the properties of synchronous automata~\cite{citeulike:10289197} and the language equivalence problem for deterministic multi-tape automata~\cite{TJ-equivalence-dtn}.
See~\cite{CAF-survey} for a detailed survey of multi-tape automata, and \cite{choff} for a historical perspective. % fundamental properties.

Khoussainov and Nerode~\cite{KhoussainovN94} introduced a framework for the presentation of first-order structures based on multi-tape automata; while~\cite{KhoussainovN94} also defines asynchronous automata, all its results target synchronous automata---and so did most of the research in this line (e.g., \cite{BlumensathG00,Rubin08,IbarraT12-CIAA,IbarraT12-TCS}).
To our knowledge, there exist only a few applications that use asynchronous multi-tape automata.
Motivated by applications in computational linguistic, \cite{ChamparnaudGKN08} discusses composition algorithms for weighted multi-tape automata.
Our intersection algorithm (Section~\ref{sec:algor-inters}) shares with \cite{ChamparnaudGKN08} the idea of accumulating delays in states; on the other hand, \cite{ChamparnaudGKN08} expresses intersection as the combination of simpler composition operations, and targets weighted automata with \emph{bounded} delays---a syntactic restriction that guarantees that reading heads are synchronized---suitable for the applications of \cite{ChamparnaudGKN08} but not for the program verification examples of Section~\ref{sec:impl-exper}.
Another application is reasoning about databases of strings (typically representing DNA sequences), for which multi-tape transducers have been used~\cite{GrahneNU99}.

Much recent research targeted the invention of decision procedures for expressive first-order fragments useful in reasoning about functional properties of programs.
Interpreted theories supporting operations on words, such as some of the examples in the present papers, include theories of arrays \cite{BMS06-vmcai,HIV08-lpar}, strings \cite{KGGHE09}, multi-sets~\cite{KPSW10}, lists~\cite{WiesMK12}, and sequences~\cite{Fur10-ATVA10}.
All these contributions (with the exception of \cite{HIV08-lpar}) use logic-based techniques, but automata-theoretic techniques are ubiquitous in other areas of verification---most noticeably, model-checking~\cite{VW86}.
The present paper has suggested another domain where automata-theoretic techniques can be useful.

\paragraph{Acknowledgements.}
Thanks to St\'ephane Demri for suggesting looking into automatic 
theories during a chat at ATVA 2010; and to Cristiano Calcagno for stimulating discussions.
Many thanks to the reviewers of several conferences, and in particular to the anonymous Reviewer~1 of the 9th International Computer Science Symposium in Russia (CSR 2014), who pointed out---in an extremely detailed and insightful review of a previous version of this paper---some non-trivial errors about the completeness analysis for the algorithm of Section~\ref{sec:positive-results}.

% \bibliographystyle{plain}
% \bibliography{mtape-bib}

%% FOR ARXIV

\clearpage
\newpage
\appendix

\section{Multi-Tape Automata: Negative Results (Section~\ref{sec:negative-results})}

While Theorem~\ref{th:not-semidecidable} subsumes the undecidability of the rational intersection problem, we can give independent proofs of two variants of the problem.
The first one uses a reduction from Post's correspondence problem; the second one, given later, a reduction from the disjointness problem for multi-tape automata.

\begin{theorem} \label{th:undec-async-intersect}
The rational intersection problem is undecidable.
\end{theorem}
\begin{proof}
We prove undecidability by reduction from Post's correspondence problem \linebreak(PCP): given a finite set 
$
\left\{
\langle x_1, y_1\rangle,
\ldots,
\langle x_m, y_m\rangle
\right\}\,,
$
of 2-words over $\Sigma$ (with $|\Sigma| \geq 2$) 
determine if there exists a sequence $i_1, i_2, \ldots, i_k$ of indices from $1, \ldots, m$ (possibly with repetitions) such that
$
x_{i_1} x_{i_2} \cdots x_{i_k} 
\ =\ 
y_{i_1} y_{i_2} \cdots y_{i_k}
$.

Given an instance of PCP, define $X = \{x_1, \ldots, x_m\}$ and $Y = \{y_1, \ldots, y_m\}$.
Assume, without loss of generality, that the symbols $1, \ldots, m$ and a marker $\#$ are not in $\Sigma$.
Consider the two languages $L_1, L_2$ defined as:
\begin{align*}
L_1 &= \{ \langle i_1\cdots i_{\alpha}, x_{i_1}\cdots x_{i_{\alpha}}, y_{i_1}\cdots y_{i_{\alpha}} \:\#\: \widetilde{x}  \rangle \mid \alpha \geq 0 \text{ and } \widetilde{x} \in X^* \}\,, \\
L_2 &= \{ \langle j_1\cdots j_{\beta}, y_{j_1}\cdots y_{j_{\beta}}, \widetilde{y} \:\#\: x_{j_1}\cdots x_{j_{\beta}}  \rangle \mid \beta \geq 0 \text{ and } \widetilde{y} \in Y^* \}\,,
\end{align*}
where the $i_k$'s and $j_k$'s are indices from $1, \ldots, m$.
It is not difficult to see that $L_1$ and $L_2$ are rational languages.
An automaton accepting $L_1$ works as follows: for each element $i_k$ on the first tape, it checks that the corresponding $x_{i_k}$ and $y_{i_k}$ respectively appear on the second and third tape; finally, it checks that only elements in $X$ appear after the $\#$ on the third tape.
An automaton for $L_2$ can follow a similar logic.

The intersection $L = L_1 \cap L_2$ consists of all words of the form
\[
\langle (i_1 \cdots i_k)^n, x_{i_1} \cdots x_{i_k}, (y_{i_1} \cdots y_{i_k})^n \# (x_{i_1} \cdots x_{i_k})^n \rangle
\]
for $n \geq 0$, such that $i_1, \ldots, i_k$ is a solution of the PCP.
If the PCP has no solution, then $L$ is the singleton $\langle \epsilon, \epsilon, \epsilon \rangle$ which is clearly rational; conversely, if the PCP has no solution, $L$ contains infinitely many words but is not rational, because its projection onto the third component has the form $u^n \# v^n$, which is non-regular~\cite{CAF-survey}.
\qedhere
\end{proof}

The proof of Theorem~\ref{th:undec-async-intersect} uses 3-words, which implies that the result carries over to any number of tapes $n \geq 3$; is it possible to generalize to $n \geq 2$?
It seems difficult to simultaneously express the PCP solution requirements and the non-regularity of one of the components.
However, a slightly weaker (but practically as useful) undecidability result for $n \geq 2$ tapes follows easily from the undecidability~\cite{CAF-survey} of the disjointness problem for rational languages (that is, determining whether the intersection $L_1 \cap L_2$ of two rational languages is empty).
We can prove that the following problem $P$ is undecidable: \emph{constructively} determine whether the intersection $L_1 \cap L_2$ of two rational languages $L_1, L_2$ is rational; ``constructively'' refers to the fact that we require that, if $L_1 \cap L_2$ is rational, then we can build an automaton $A_{1,2}$ such that $\lang{A_{1,2}} = L_1 \cap L_2$.
Assume, to the contrary, that $P$ is decidable.
Then, we have a decision procedure for the disjointness problem: if $L_1 \cap L_2$ is rational, construct and automaton $A_{1,2}$ that accepts it, and test $A_{1,2}$ for emptiness; otherwise, $L_1 \cap L_2$ is not rational, and hence certainly $L_1 \cap L_2 \neq \emptyset$.

\pagebreak
\section{Under-Approximation of Intersection (Section~\ref{sec:algor-inters})}
\label{app:under-appr-inters}

\begin{figure}[!h]
\begin{lstlisting}
async_next (D, q): SET [$\langle q', h_1, \ldots, h_t \rangle$]
   -- $q$ is always reachable from itself
   Result $:= \{ \langle q, \epsilon, \ldots, \epsilon\rangle \}$
   -- for every tape other than $q$'s
   for each $t_i \in \{t_1^D, \ldots, t_t^D\} \setminus \tau^D(q)$ do
      $P := \text{all shortest paths } p \text{ from }q \text{ to some }\overline{q}  \text{ such that: }$
            $\tau^D(\overline{q}) = t_i \text{ and no state } \widetilde{q} \text{ with } \tau^D(\widetilde{q}) = t_i \text{ appears in } p \text{ before } \overline{q}$
      -- each element in $P$ is a sequence of transitions
      for each $e_1\,\cdots\,e_m \in P$ do
         $h_1, \ldots, h_t := \epsilon$
         -- each transition is a triple (source, input, target)
         for each $(q_1, \sigma, q_2) \in e_1\,\cdots\,e_m$ do
            -- add the transition to the sequence corresponding
            -- to its source's tape
            $h_{\tau^D(q_1)} := h_{\tau^D(q_1)} + (q_1, \sigma, q_2)$
         -- $q_2(e_m)$ is the target state of the last transition $e_m$
         Result $:=$ Result $\cup \langle q_2(e_m), h_1, \ldots, h_t\rangle$

new_states (P: SET[$\langle p, h_1, \ldots, h_m \rangle$], Q: SET[$\langle q, k_1, \ldots, k_n \rangle$]): S
   S $:= \emptyset$
   for each $\langle p, h_1, \ldots, h_m \rangle \in P$, $\langle q, k_1, \ldots, k_n \rangle \in Q$ do
      -- if delays on synchronized tapes are consistent
      if $\forall i \in T^A \cap T^B:$ cons($h_i, k_i$) then
         for each $t \in T$ do $S := S \cup \{\langle p, q, t, h_1, \ldots, h_m, k_1, \ldots, k_n \rangle \}$ end
   -- Here $Q$ denotes $C$'s set of states, not the input argument
   for each $r \in S$ do if $r \not\in Q$ then s.push (r) end

compose_transition (P: SET[$(p, h_1, \ldots, h_m)$], Q: SET[$(q, k_1, \ldots, k_n)$], 
                    d: ($h_1, \ldots, h_m, k_1, \ldots, k_n$), $\sigma$, r)
   $J_A := \{ (p, h_1\,h_1', \ldots, h_m\,h_m') \mid (p, h_1', \ldots, h_m') \in P \}$
   $J_B := \{ (q, k_1\,k_1', \ldots, k_n\,k_n') \mid (q, k_1', \ldots, k_n') \in Q \}$
   $S :=$ new_states ($J_A, J_B$)
   for each $r' \in S$ do $\delta := \delta \cup \{r, \sigma, r'\}$ end
\end{lstlisting}
\caption{Routines \lstinline|async_next|, \lstinline|new_states|, \lstinline|compose_transition|.}
\label{fig:composeTransition}
\label{fig:newStates}
\label{fig:asyncNext}
\end{figure}

\begin{figure}[!h]
\begin{lstlisting}
intersect (max_states, max_delay)
  $Q := \emptyset$ ; $s := \emptyset$
  -- Initially reachable states
  $J_A := \bigcup_{i \in Q_0^A}$ async_next (A, i)  ;  $J_B := \bigcup_{i \in Q_0^B}$ async_next (B, i)
  S $:=$ new_states ($J_A$, $J_B$)  ;   $Q_0 := S$
  until $s = \emptyset$ or $|Q| \geq$ max_states loop
    r $:=$ ($q_a, q_b, t, h_1, \ldots, h_m, k_1, \ldots, k_n$) $=$ s.pop  
    if $\forall d \in \{ h_1, \ldots, k_n \}: |d| \leq $ max_delay then $Q := Q \cup \{ r \}$ else continue
    if $t \in T^A \cap T^B$ then -- event on shared tape
      if $h_t = (u_a, \sigma, u_a') \overline{h_t}$ and $k_t = (u_b, \sigma, u_b') \overline{k_t}$ then
        -- delayed transition on both $A$ and $B$
        $P :=$ async_next $(A, q_a)$ ; $Q :=$ async_next $(B, q_b)$
        $d := (h_1, \ldots, \overline{h_t}, \ldots, h_m, k_1, \ldots, \overline{k_t}, \ldots, k_n)$
        compose_transition ($P, Q, d, \sigma, r$)
      elseif $h_t = (u_a, \sigma, u_a') \overline{h_t}$ and $k_t = \epsilon$ then
        -- delayed transition on $A$, normal transition on $B$
        $P :=$ async_next $(A, q_a)$
        $Q := \{$ async_next ($B, q_b'$) $\mid (q_b, \sigma_b, q_b') \in \delta^B \land \sigma = \sigma_b \land \tau^B(q_b) = t \}$
        $d := (h_1, \ldots, \overline{h_t}, \ldots, h_m, k_1, \ldots, k_n)$
        compose_transition ($P, Q, d, \sigma, r$)
      elseif $h_t = \epsilon$ and $k_t = (u_b, \sigma, u_b') \overline{k_t}$ then
        -- delayed transition on $B$, normal transition on $A$
        $\cdots$
      elseif $h_t = k_t = \epsilon$ then
        -- normal transition on both $A$ and $B$
        for each $\sigma \in \Sigma$ do
          $P := \{$ async_next ($A, q_a'$) $\mid (q_a, \sigma_a, q_a') \in \delta^A \land \sigma_a = \sigma \land \tau^A(q_a) = t \}$
          $Q := \{$ async_next ($B, q_b'$) $\mid (q_b, \sigma_b, q_b') \in \delta^B \land \sigma_b = \sigma \land \tau^B(q_b) = t \}$
          $d := (h_1, \ldots, h_m, k_1, \ldots, \ldots, k_n)$
          compose_transition ($P, Q, d, \sigma, r$)
    elseif $t \in T^A \setminus T^B$ then -- event on $A$'s non$\text{-}$shared tape
      if $h_t = (u_a, \sigma, u_a') \overline{h_t}$ then -- delayed transition on $A$, $B$ stays
        $P :=$ async_next $(A, q_a)$ ; $Q := \{ (q_b, \epsilon, \ldots, \epsilon) \}$
        $d := (h_1, \ldots, \overline{h_t}, \ldots, h_m, k_1, \ldots, k_n)$
        compose_transition ($P, Q, d, \sigma, r$)
      elseif $h_t = \epsilon$ then -- normal transition on $A$, $B$ stays
          $Q := \{ (q_b, \epsilon, \ldots, \epsilon) \}$
          for each $\sigma \in \Sigma$ do
            $P := \{$ async_next ($A, q_a'$) $\mid (q_a, \sigma_a, q_a') \in \delta^A \land \sigma_a = \sigma \land \tau^A(q_a) = t \}$
            $d := (h_1, \ldots, h_m, k_1, \ldots, k_n)$
            compose_transition ($P, Q, d, \sigma, r$)
    elseif $t \in T^B \setminus T^A$ then -- event on $B$'s non$\text{-}$shared tape
      $\cdots$
\end{lstlisting}
\caption{Routine \lstinline|intersect|.}
\label{fig:intersect}
\end{figure}

\clearpage

\section{Correctness and Completeness (Section~\ref{sec:correct-precise})}

\begin{lemma}[Pumping lemma] \label{lemma:pumping}
Let $L$ be an $n$-rational language. 
Then there exists an integer $N \geq 1$ such that every word $\langle x_1, \ldots, x_n \rangle \in L$ where $|x_1| + \cdots + |x_n| \geq N$ can be written as $\langle p_1q_1r_1, \ldots, p_nq_nr_n \rangle$, with $q_k \neq \epsilon$ for at least one $1 \leq k \leq n$, and $\langle p_1 q_1^m r_1, \ldots, p_nq_n^mr_n\rangle$ is in $L$ for every $m \in \naturals$.
\end{lemma}
\begin{proof}
Let $A_L$ be an automaton accepting $L$; then, the number of states of $A_L$ is the pumping length $N = M$.
Consider a word $w = \langle x_1, \ldots, x_n \rangle \in L$ with length $|x_1| +\, \cdots\, + |x_n| \geq N$.
A computation accepting $w$ visits $N + 1$ states of $A_L$; by the pigeonhole principle, there exists a state $s$ in the sequence which is visited twice.
The sequence of symbols read in the transitions that go from the first to the second visit of $s$ determines an $n$-word $\langle q_1, \ldots, q_n \rangle$ with at least one $q_k \neq \epsilon$.
Looping an arbitrary number of times over the sequence that starts and ends on $s$  determines words that are all accepted by $A_L$, and hence belong to $L$.
\qedhere
\end{proof}

\section{Asynchronous Rational Theories} \label{sec:asynchr-autom-theor}
The \emph{signature} $S_\Theta = C \cup F \cup R$ of a first-order theory $\Theta$ is a set of constant $C$, function $F$, and predicate $R$ symbols.
A \emph{quantifier-free formula} of $\Theta$ is built from constant, function, and predicate symbols of $S_\Theta$, as well as variables $x, y, z, \ldots$ and logical connectives $\limpl, \lor, \land, \neg$.
An \emph{interpretation}\footnote{For simplicity, we do not discuss how to \emph{axiomatize} the semantics of interpreted items.} $I_\Theta$ assigns constants, functions, and predicates over a domain $D$ to each element of $C$, $F$, and $R$.
It is customary that $R$ include an equality symbol $=$ with its natural interpretation.
Then, assume without loss of generality that $\Theta$ is \emph{relational}, that is $F = \emptyset$; to this end, introduce a $(m+1)$-ary predicate $R_f$ for every $m$-ary function $f$ such that $R_f(x_1, \ldots, x_m, y)$ holds iff $f(x_1, \ldots, x_m) = y$.
A \emph{model} $M$ of a formula $F$ of $\Theta$ is an assignment of values to the variables in $F$ that is consistent with $I_\Theta$ and makes the formula evaluate to true; write $M \models F$ to denote that $M$ is a model of $F$.
The set of all models of a formula $F$ under an interpretation $I_\Theta$ is denoted by $\allmodels{F}{I_\Theta}$.
$F$ is \emph{satisfiable} in the interpretation $I_\Theta$ if $\allmodels{F}{I_\Theta} \neq \emptyset$; it is \emph{valid} if $\allmodels{F}{I_\Theta}$ contains all variable assignments that are consistent with $I_\Theta$.

Similar to automatic presentations, a \emph{rational presentation} of a first-order theory $\Theta$ consists of:
\begin{enumerate}
\item A finite alphabet $\Sigma$;
\item A surjective mapping $\nu: S \to D$, with $S$ a regular subset of $\Sigma^*$, that defines an encoding of elements of the domain $D$ in words over $\Sigma$;
\item A 2-tape automaton $\autd{\mathsf{eq}}$ whose language is the set of 2-words $\langle x,y \rangle \in (\Sigma^*)^2$ such that $\nu(x) = \nu(y)$;
\item For each $m$-ary relation $R_m \in R$, an $m$-tape automaton $\autd{R_m}$ whose language is the set of $m$-words $\langle x_1,\ldots, x_m \rangle \in (\Sigma^*)^m$ such that $R_m(\nu(x_1), \ldots, \nu(x_m))$ holds.
\end{enumerate}
A first-order theory with rational presentation is called \emph{rational theory}.
If the automata of the presentation are deterministic (resp.\ synchronous, asynchronous) the theory is also called deterministic (resp.\ synchronous, asynchronous).

\begin{example}[Rational theory of concatenation]
The theory of concatenation over $\{a,b\}^*$ is the first-order theory with constant $\epsilon$ (the empty sequence), sequence equality $=$, and concatenation predicate $R_{\cat}$ such that $R_{\cat}(x,y,z)$ holds iff $z$ is the concatenation of $x$ and $y$.
This theory is asynchronous rational, with $\Sigma = \{a,b\}$, $\nu$ the identity function, $\autd{\mathsf{eq}}$ as in Figure~\ref{fig:Aeq}, and $\autd{R_{\cat}}$ as in Figure~\ref{fig:Acat}.
\end{example}

Consider a quantifier-free formula $F$ of a rational theory $\Theta$.
To decide if $F$ is satisfiable we can proceed as follows.
First, build an automaton $\autd{F}$ that recognizes exactly the models of $F$.
This is done by composing the elementary automata of the theory according to the propositional structure of $F$; namely, for sub-formulas $G, H$, negation $\neg G$ corresponds to complement $\comp{\autd{G}}$, disjunction $G \lor H$ corresponds to union $\autd{G} \cup \autd{H}$, and conjunction $G \land H$ corresponds to intersection $\autd{G} \cap \autd{H}$.
To verify if $F$ is valid, test whether $\autd{\neg F} = \comp{\autd{F}}$ is empty: $\lang{\autd{\neg F}}$ is empty iff $\neg F$ is unsatisfiable iff $F$ is valid.

We can apply this procedure only when the automaton $\autd{F}$ is effectively constructible, which is not always the case for asynchronous rational theories because asynchronous automata lack some closure properties (see Section~\ref{sec:clos-prop-decid})---intersection, in particular.
The following section, however, shows some non-trivial examples of formulas whose rational presentation falls under the criterion of Corollary~\ref{lem:completeness} (and whose components to be complemented are deterministic), hence we can decide their validity by means of automata constructions.

\section{Implementation and Experiments (Section~\ref{sec:impl-exper})}
\label{app:impl-exper-sect}

\begin{align*}
\vc_0 & \equiv\ 
|y| \geq m \land y = \rest(x)
\;\limpl\;
|x| \geq n \land m = n - 1 \\
\vc_1 & \equiv\  |y| \geq m \land y = \rest(x) \;\limpl\; |x| \geq n \land m = n - 1 \\
\mathsf{cat}_0 & \equiv\ 
x \cat y = z \land \mathit{last}(z) = u \land \mathit{last}(y) = v \;\limpl\; u = v \\
\ce_1 & \equiv\  |y| \geq m \land y = \rest(x) \;\limpl\; |x| < n \\
\ce_2 & \equiv\  
|\Result| = u \land u = |y| - m \land y = \rest(x) 
\;\limpl\; 
|\Result| = v
\end{align*}
\begin{multline*}
\vc_2 \equiv\ 
|\Result| = u \land u = |y| - m \land y = \rest(x)
\\ \;\limpl\;
|\Result| = v \land v = |x| - n \land m = n-1 \land |x| = n
\end{multline*}

%%% Local Variables: 
%%% mode: latex
%%% TeX-master: "mtape"
%%% End: 

\end{document}